\documentclass[10pt,a4paper]{article}
\usepackage[latin1]{inputenc}
\usepackage{amsmath}
\usepackage{amsfonts}
\usepackage{amssymb}
\usepackage{graphicx}
\usepackage{grffile}
\usepackage{amsmath,amssymb,amsbsy}
\usepackage{hyperref}

\usepackage{color}
\usepackage{comment}
\usepackage{lipsum}
\usepackage{graphicx}%
\usepackage{subfigure}
\usepackage[normalem]{ulem}
\usepackage{float}
\usepackage{soul}
\usepackage{authblk}
\newtheorem{theorem}{Theorem}
\newtheorem{lemma}[theorem]{Lemma}
	
\newtheorem{proposition}[theorem]{Proposition}
\newtheorem{corollary}[theorem]{Corollary}	
\newtheorem{remark}[theorem]{Remark}	
\newtheorem{problem}[theorem]{Problem}

\newcommand{\kar}[1]{{\color{blue} #1}}
\allowdisplaybreaks

\makeatletter
\DeclareRobustCommand{\qed}{%
	\ifmmode 
	\else \leavevmode\unskip\penalty9999 \hbox{}\nobreak\hfill
	\fi
	\quad\hbox{\qedsymbol}}
\newcommand{\openbox}{\leavevmode
	\hbox to.77778em{%
		\hfil\vrule
		\vbox to.675em{\hrule width.6em\vfil\hrule}%
		\vrule\hfil}}
\newcommand{\qedsymbol}{\openbox}
\newenvironment{proof}[1][\proofname]{\par
	\normalfont
	\topsep6\p@\@plus6\p@ \trivlist
	\item[\hskip\labelsep\itshape
	#1.]\ignorespaces
}{%
	\qed\endtrivlist
}
\newcommand{\proofname}{Proof}
\makeatother

\begin{document}
	\title{PDE-Based Optimization for Stochastic Mapping and Coverage Strategies using Robotic Ensembles}
\author[1]{Karthik Elamvazhuthi}
\author[2]{Hendrik Kuiper}            
\author[1]{Spring Berman}
\affil[1]{School for Engineering of Matter, Transport and Energy, Arizona State University, Tempe, AZ, USA}

\affil[2]{School of Mathematical and Statistical Sciences, Arizona State University, Tempe, AZ, USA}

	\maketitle
	
	\begin{abstract}
		This paper presents a novel partial differential equation (PDE)-based framework for controlling an ensemble of robots, which  have limited sensing and actuation capabilities and exhibit stochastic behaviors, to perform mapping and coverage tasks.  We model the ensemble population dynamics as an advection-diffusion-reaction PDE model and formulate the mapping and coverage tasks as identification and control problems for this model. In the mapping task, robots are deployed over a closed domain to gather data, which is unlocalized and independent of robot identities, for reconstructing the unknown spatial distribution of a region of interest.  We frame this task as a convex optimization problem whose solution represents the region as a spatially-dependent coefficient in the PDE model. We then consider a coverage problem in which the robots must perform a desired activity at a programmable probability rate to achieve a target spatial distribution of activity over the reconstructed region of interest.  We formulate this task as an optimal control problem in which the PDE model is expressed as a bilinear control system, with the robots' coverage activity rate and velocity field defined as the control inputs.  We validate our approach with simulations of a combined mapping and coverage scenario in two environments with three target coverage distributions.
		\end{abstract}

	\section{Introduction}
	
	%
	
	
	
	Over the past few decades, partial differential equation (PDE) models of multi-agent systems have been used extensively in mathematical biology to analyze collective behaviors such as chemotaxis, flocking, schooling, predator-prey interactions, and pattern formation  \cite{okubo1986dynamical}.  Many of these models are linear or nonlinear advection-diffusion type PDEs, which describe the spatiotemporal evolution of probability densities of agents. Mathematical tools such as bifurcation analysis, optimization, and control theory can be applied to these continuum {\it macroscopic models} to make qualitative and quantitative predictions about the system behavior.  
	Typically, each PDE model corresponds to a discrete {\it microscopic model} that captures the stochastic and deterministic actions of individual agents. While these microscopic models are more accurate descriptions of the agents' behavior, the macroscopic models enable tractable analysis for large agent numbers.
	
	
	
	Recently, this work has motivated the use of similar types of PDEs to model and control the spatiotemporal dynamics of very large collectives, or {\it swarms} \cite{Francesca2016}, of small, resource-constrained robots (e.g., \cite{Karydis2016,Sitti2015}) that are currently being developed for applications such as environmental monitoring, exploration, surveillance, disaster response, and biomedical procedures.  PDEs have been used to characterize the distributions of chemotactic robots in a diffusive fluid environment \cite{galstyan2005modeling}, miniature robots inspecting a model of jet turbine blades \cite{prorok2011multi}, and honeybee-inspired agents that aggregate at the optimal value of a scalar field \cite{correll2015probabilistic}.  The parameters of these PDE models can be mapped to control inputs that drive the robots' motion and probability rates of switching between states or tasks, and the collective behavior of the robots follows the PDE model prediction in expectation.  
	Several works have exploited this correspondence to control the spatial distribution of a ensemble \cite{milutinovic2006modeling,foderaro2014distributed}.  
	These control approaches can be viewed as extensions of stochastic task allocation schemes based on nonspatial rate equation models \cite{BermanTRO2009,correll2006system,Martinoli04,lerman2006analysis}.  Other applications of continuum population dynamical models to multi-agent control include optimized confinement strategies \cite{haque2014efficient}, consensus using the theory of mean field controls \cite{nourian2013nash}, controlled flocking \cite{piccoli2014control} that includes non-parallel motions \cite{hanstyled}, and pattern generation in the presence of obstacles \cite{pimenta2008control}. There has also been some recent work on using PDEs to model Laplacian network dynamics of agents for formation control; see \cite{frihauf2011leader,meurer2011finite,ElamvazhuthiDARS2014} and references therein. 
	
	We apply this PDE-based modeling framework to develop a control approach for  allocating tasks among an ensemble of robots.  
	In our scenarios, a {\it task} is defined as a desired activity that a robot performs in a certain spatial region of the environment. The tasks can be performed in parallel, and multiple robots can be simultaneously allocated to each task. While various deterministic approaches have been developed for multi-robot task allocation, including centralized and decentralized market-based techniques~\cite{Dias2006,choi-tro09} and centralized methods for optimal task assignment and trajectory planning \cite{adler2015efficient,turpin2014capt}, their computation and/or communication requirements do not scale well to very large numbers of robots and tasks.
	In contrast to these works, we develop a {\it stochastic} approach in which tasks are performed at random times by unidentified robots with limited computing capabilities and no communication or global localization. Such limitations will be common in swarm robotic platforms, e.g. micro aerial vehicles \cite{Karydis2016} and microrobots \cite{Sitti2015}, and in scenarios where the robots operate in GPS-denied environments where communication is impractical or unreliable. In our proposed approach, a task allocation emerges from the collective ensemble activity.
	
	We first consider a {\it mapping} problem in which the objective is to estimate a scalar spatial field from unlocalized data obtained by the robots. We then define a {\it coverage} problem in which the ensemble must produce a target spatial density of activity over a region of interest, which may be estimated in the {\it mapping} problem. For this problem, we express the PDE model as a {\it bilinear control system} \cite{ball1982controllability} and formulate an optimal control problem that computes the control inputs. Since we do not assume that agents are capable of global localization or estimation of the local agent population density,
	we frame the coverage problem as an open-loop control problem that does not require feedback on agent positions or densities. 
	We follow the variational approach described in \cite{troltzsch2010optimal} for optimal control of the PDE model. While there has been some prior work on bilinear optimal control of systems of PDEs \cite{ou2010stability,annunziato2014optimal}, these works do not address the types of PDEs that we consider. An optimal control problem for a bilinear parabolic PDE was formulated in \cite{ou2010stability} with the diffusion coefficient as the control. In \cite{annunziato2014optimal}, bilinear control of a class of advection-reaction systems was considered; unlike our PDE models, these systems did not include diffusion.
	
	The mapping and controller synthesis approaches described in this paper require a central supervisor with the computational capabilities necessary to solve the associated optimization problems. Despite this centralized component, the approaches are scalable with the number of agents in the ensemble since each agent executes the same controllers with the same control variables, which are preprogrammed or broadcast by the supervisor. 
	In our coverage strategy, there are only three control variables to be computed; in contrast, the most naive approach to controlling an ensemble of $N$ agents moving in $d$ dimensions would require the computation of $N^d$ control inputs.
	
	

	
	
	We first presented our {\it coverage} approach in \cite{elamvazhuthi2015optimal}, where we introduced a similar optimal control problem, derived the gradient of the objective functional with respect to the control parameters, and used a gradient descent algorithm to compute the optimal control.  This paper provides a complete analysis of our approach in \cite{elamvazhuthi2015optimal} by investigating the well-posedness of the PDE model and the optimal control problem.  The theory of weak solutions that we use to establish the well-posedness of the PDE model is classical \cite{evans1998partial}.  However, to the best of our knowledge, there have been no prior results on well-posedness that can be directly applied to our  model, which is a system of PDEs in which diffusion is present only in one of the species, the control variables are time-dependent, and a zero-flux boundary condition is imposed on the boundary of a Lipschitz domain. In this paper, we prove the existence and uniqueness of solutions of our PDE model by deriving suitable energy estimates for the solutions.
	We also use these derived energy estimates to ensure that the computation of the gradient, performed using the adjoint equation approach, is well-posed.  Moreover, we prove the existence of an optimal control for the problem using standard compactness arguments adapted to the PDE control setting \cite{troltzsch2010optimal}.  In addition to this analysis, our formulation of the {\it mapping} problem in the same framework is a novel contribution of this paper; in \cite{elamvazhuthi2015optimal} it was assumed that the environment is known beforehand.
	
	The paper is organized as follows. Section \ref{sec:TskDsc} describes the robot capabilities and their programmed behaviors during the mapping and coverage  assignments, 
	and  Section \ref{sec:modsec} defines the microscopic and macroscopic models of the ensemble and its activity during each  assignment. Section \ref{sec:mathprelim} defines key mathematical terminology that is used in Sections \ref{sec:optMap} and \ref{sec:optCoverage} to formulate and analyze the mapping and coverage objectives, respectively, as optimization problems that incorporate the macroscopic models.  We validate our approach in Section \ref{sec:simulat} with simulations in which a region of interest must first be mapped and then covered with a target distribution of robot activity, and we conclude in Section \ref{sec:conc}. 
	
	
	
	

	\section{Task Descriptions and Assumptions}
	\label{sec:TskDsc}
	We consider a scenario in which $\mathbf{(1)}$ a small number of agents must map a region of interest in an unknown, bounded environment, which we refer to as the {\it mapping assignment}, and then $\mathbf{(2)}$ a larger ensemble of agents must produce a target spatial distribution of activity within the mapped regions, which we call the {\it coverage assignment}.  For instance, this activity could consist of sensor measurements, or as in our previous work \cite{elamvazhuthi2015optimal}, contacts with flowers to effect crop pollination. The mapping and coverage assignments will be formulated in a decoupled manner by posing them as two separate optimization problems in terms of their associated mean-field PDE models. We will then demonstrate through numerical simulations that these two problems can be solved sequentially in order to achieve the desired coverage objective.
	
	\subsection{Robot capabilities}
	We assume that the agents lack global localization, inter-agent communication, and prior data about the environment.  Each agent is equipped with local sensing capabilities, allowing it to detect and distinguish between different types of regions within its sensing range, and a compass, enabling it to move with a specified heading.  Additionally, the agents have sufficient memory to store the times at which they record observations of regions of interest. Similarly, it is assumed that agents have sufficient memory to store time-dependent velocity and task-switching parameters for the coverage assignment. 
	
	Both the mapping and coverage assignments involve numerical optimization computations that are performed offline by an external supervisor.  After the mapping assignment, the supervisor collects recorded information from the robots and uses this information to reconstruct the map of the environment. We note that the supervisor does not have information on the individual identities of the robots.  Based on this map, the supervisor calculates the parameters of the agent behaviors for a specified coverage objective and broadcasts these parameters to the agents before they are deployed for the coverage assignment. This {\it broadcast  architecture} has been previously proposed for controlling large numbers of robots \cite{michael2010architecture}. Although the external supervisor constitutes a potential single point of failure, it enables the control of large ensembles of agents that are subject to the constraints that we consider.
	
	
	
	
	\subsection{Robot controllers} \label{sec:ControlPolicies}
	The agents traverse the environment, a bounded open connected set $\Omega \subset \mathbb{R}^2$ with Lipschitz continuous boundary $\partial\Omega$, with a combination of deterministic and stochastic motion during a deployment. Each agent moves with a time-dependent velocity field $\mathbf{v}(t) \in \mathbb{R}^2$, which may be broadcast to the agents or imposed on them using external stimuli, such as magnetic fields in microrobotic applications.  Concurrently, each agent exhibits random motion that may be programmed, for instance to perform probabilistic search and tracking, or that arises from inherent sensor and actuator noise.  We assume that this random movement can be modeled as a Brownian motion that drives diffusion with an associated diffusion coefficient $D$. This approach to modeling the motion of members of a robotic swarm has, for example, been used previously in \cite{hamann2008framework}.
	
	The agents switch stochastically between behavioral states at constant or time-dependent {\it transition rates}, which define the probability per unit time for an agent to switch from one state to another.  We define a {\it region of interest}, which may be disconnected and is assumed to be Lebesgue measurable, as $\Gamma \subset \Omega$.  During {\it mapping}, a continually moving agent records observations of $\Gamma$ at a constant rate $k_o$.  During {\it coverage}, an agent that is moving over $\Gamma$ pauses to perform an action (such as a sensor measurement) at a time-dependent rate $k(t)$, and it resumes moving at a constant rate $k_f$, which determines the time taken to complete the action.  This coverage strategy can be extended to scenarios where there are different types of regions, and the agents perform actions over each region at a different rate, as in our prior work \cite{elamvazhuthi2015optimal}.  
	
	
	We specify that the agents' velocity field and transition rates are controllable parameters. In the {\it mapping} assignment, $\mathbf{v}(t)$ is designed to ensure thorough exploration of the environment (for instance, following a lawnmower pattern), and $k_o$ is assigned a high value to yield a large number of observations and thus produce an accurate map.  In the {\it coverage} assignment, $k_f$ is fixed while $\mathbf{v}(t)$ and $k(t)$ are computed prior to the agents' deployment according to the optimal control approach in Section \ref{sec:optCoverage}.

	\section{Models of Ensemble Mapping and Coverage}
	\label{sec:modsec}
	\subsection{Microscopic Models} \label{sec:micromodel}
	
	The microscopic model is used to simulate the individual robots' motion and probabilistic decisions that are produced by the robot controllers in Section \ref{sec:ControlPolicies}. 
	
	\vspace{2mm}
	
	\begin{remark}
		Here we describe the microscopic model at a formal level as a discrete-time stochastic process. A rigorous correspondence between the microscopic model and the macroscopic model for the case $\Omega = \mathbb{R}^2$ was shown recently by the authors and collaborators in \cite{zhang2017performance}.
	\end{remark}
	
	We model a robot's changes in state and performance of desired actions as a Chemical Reaction Network (CRN) in which the possible species are $M$, a {\it moving} robot; $S$, a {\it stationary} robot; and $A$, an instance of a desired robot activity.  These reactions can only occur when a robot is located in the region of interest $\Gamma$. For example, in an artificial pollination scenario \cite{elamvazhuthi2015optimal}, $\Gamma$ could represent the subset of the domain in which flowers are distributed. A robot's {\it mapping} activity consists of a single irreversible reaction,
	\begin{equation}
	M \hspace{2mm} \xrightarrow{k_{o}} \hspace{2mm} M + A, \label{eq:CRN_map}
	\end{equation}  
	where $A$ is the robot's observation of the region of interest. A robot's {\it coverage} activity consists of two irreversible reactions,
	\begin{eqnarray}
	 M ~~& \xrightarrow{k(t)} & ~~S + A  \label{eq:CRN1_pa}, \\
	 S ~~& \xrightarrow{k_{f}} & ~~M  \label{eq:CRN2_pa},
	\end{eqnarray}  
	where $A$ is a desired action that the robot performs.
	
	In addition, we model the displacement of a robot $i$ over a timestep $\Delta t$ using the standard-form Langevin equation \cite{Gillespie2000},
	\begin{equation}
	\mathbf{X}_i(t+\Delta t) - \mathbf{X}_i(t) = Y_i(t)\big (\mathbf{v}(t) \Delta t + (2 D \Delta t)^{1/2}~\mathbf{Z}_i(t) \big), \label{eq:RWPTpos}
	\end{equation}
	
	where $\mathbf{X}_i(t) \in \mathbb{R}^2$ is the position of  robot $i \in \lbrace 1,..,N\rbrace$ at time $t$ and $\mathbf{Z}_i(t) \in \mathbb{R}^2$ is a vector of independent standard normal random variables that are generated at time $t$.  Here, $Y_i(t)$ is a binary variable associated with the discrete state of robot $i$ at time $t$. When the robot is performing the mapping assignment, $Y_i(t) = 1 $ for all $t$. When it is performing the coverage assignment, $Y_i(t)=1$ if the robot is in state $M$ at time $t$ and $Y_i(t) = 0$ if it is in state $S$ at time $t$. During the coverage assignment, $Y_i(t)$ evolves according to the conditional probabilities $\mathbb{P}(Y_i(t+\Delta t)= 0 ~|~ Y_i(t) = 1) =H_{\Gamma}(\mathbf{X}_i(t))k(t)\Delta t$ and $\mathbb{P}(Y_i(t+\Delta t)= 1 ~|~ Y_i(t) = 0) =k_f \Delta t$, where $H_{\Gamma}$ is the indicator function of the set $\Gamma$ and $\mathbb{P}$ is the probability measure induced by the process $(\mathbf{X},Y)$ on the set of sample paths. The robot avoids collisions with the domain boundary by performing a specular reflection when it encounters this boundary. We emphasize that the velocity $\mathbf{v}(t)$ and diffusion coefficient $D$ are the same for each robot $i$. This assumption enables the coverage problem to be posed as a PDE control problem, as shown in Sections \ref{sec:optMap} and \ref{sec:optCoverage}.

	\subsection{Macroscopic Models}
	\label{sec:macrosec}
	The macroscopic model consists of a set of advection-diffusion-reaction (ADR) PDEs that describe the evolution of probability densities of agents that follow the microscopic model, conditioned on the 
	initial distribution of the agents' states. Since we assume that there are no interactions between the agents, the random variables associated with the agents are independent and identically distributed, and hence we can drop the subscript $i$ from $\mathbf{X}_i(t)$, $Y_i(t)$, and $\mathbf{Z}_i(t)$ and represent the macroscopic model as a single system of PDEs rather than $N$ systems of PDEs. 
	The state variables of this model are the population density fields of {\it moving} robots, $y_1(\mathbf{x},t)$; {\it stationary} robots, denoted by $y_2(\mathbf{x},t)$ in the {\it coverage} model; and instances of the desired robot activity, denoted by $y_2(\mathbf{x},t)$ in the {\it mapping} model and by $y_3(\mathbf{x},t)$ in the {\it coverage} model.  
	We define $Q = \Omega \times (0,T) $ and $\Sigma = \partial\Omega \times (0,T) $ for some fixed final time $T$. The vector $\mathbf{n}(\mathbf{x}) \in \mathbb{R}^2$ is the outward normal to $\partial\Omega$, defined for almost every $\mathbf{x} \in \partial \Omega$.
	We also define an initial density of {\it moving} robots $y_0(\mathbf{x})$ over $\Omega$ that is normalized such that $\int_{\Omega} y_0(\mathbf{x})d\mathbf{x} = 1$.
	
	The macroscopic model of {\it mapping} is given by:
	\begin{eqnarray}
	&& \hspace{-2mm} \frac{\partial y_{1}}{\partial t} = \nabla \cdot (D \nabla y_{1}- \mathbf{v}(t) y_{1}) \hspace{5mm} in \hspace{2mm} Q, \nonumber \\
	&& \hspace{-2mm}  \frac{\partial y_{2}}{\partial t} = k_o H_{\Gamma}(\mathbf{x}) y_{1} \hspace{5mm} in \hspace{2mm} Q \label{eq:MapBCS}
	\end{eqnarray}
	with the no-flux boundary condition
	\begin{equation}
	\mathbf{n} \cdot (D \nabla y_{1} - \mathbf{v}(t)  y_{1})= 0 \hspace{2mm} on \hspace{2mm} \Sigma \label{eq:NFBC} \\
	\end{equation}
	and initial conditions
	\begin{equation}
	y_1(\cdot,0)= y_0, \hspace{2mm} y_2(\cdot,0) = \mathbf{0} \hspace{2mm} on \hspace{2mm} \Omega,
	\label{eq:ICBC}
	\end{equation}
	where $\mathbf{0}$ is the function on $\Omega$ that takes the value 0 almost everywhere. The microscopic model of mapping defined in Section \ref{sec:micromodel} is related to the above macroscopic model by the approximation
	$\mathbb{P}(\mathbf{X}(t) \in \Lambda) ~\approx~ \int_{\Lambda} y_1(\mathbf{x},t) d\mathbf{x}$. Additionally, $\int_{\Lambda} y_2(\mathbf{x},t)d\mathbf{x}$ is the expected number of observations recorded by the robots in the region $\Lambda$ up until time $t$.
	
	The macroscopic model of {\it coverage} is defined as:  
	\begin{eqnarray}
	\hspace{0mm} \frac{\partial y_{1}}{\partial t} & = & \nabla \cdot (D \nabla y_{1}- \mathbf{v}(t) y_{1}) - k(t) H_{\Gamma}(\mathbf{x}) y_{1} + k_f y_{2} \hspace{2mm} in \hspace{2mm} Q, \nonumber \\
	\hspace{0mm} \frac{\partial y_{2}}{\partial t} & = & k(t) H_{\Gamma}(\mathbf{x}) y_{1} - k_f y_{2}  \hspace{2mm} in \hspace{2mm} Q,\nonumber \\
	\hspace{0mm}  \frac{\partial y_{3}}{\partial t} & =& k(t) H_{\Gamma}(\mathbf{x}) y_{1}  \hspace{2mm} in \hspace{2mm} Q, 
	\label{eq:Macro1}
	\end{eqnarray} 
	with the no-flux boundary condition \eqref{eq:NFBC} and initial conditions
	\begin{equation} \label{eq:covIC}
	y_1(\cdot,0)= y_0, \hspace{2mm} y_2(\cdot,0) = y_3(\cdot,0)=\mathbf{0} \hspace{2mm} on \hspace{2mm} \Omega.
	\end{equation}
	
	In this macroscopic model, the density fields $y_1(\mathbf{x},t)$ and $y_2(\mathbf{x},t)$ are related to the microscopic model of coverage, given by the stochastic process $(\mathbf{X}(t),Y(t))$ defined in Section \ref{sec:micromodel}, as follows:  \begin{eqnarray}
	\mathbb{P}\left((\mathbf{X}(t),Y(t)) \in (\Lambda \times \lbrace 1 \rbrace )\right) &~\approx~& \int_{\Lambda} y_1(\mathbf{x},t) d\mathbf{x}, \nonumber \\
	\mathbb{P}\left((\mathbf{X}(t),Y(t)) \in (\Lambda \times \lbrace 0 \rbrace)\right) &~ \approx ~& \int_{\Lambda} y_2(\mathbf{x},t) d\mathbf{x}, \nonumber \\
	\end{eqnarray} 
	where $\Lambda$ is a measurable subset of $\Omega$. Additionally, $\int_{\Lambda} y_3(\mathbf{x},t)d\mathbf{x}$ represents the expected number of instances of robot coverage activity performed in the region $\Lambda$ up until time $t$. 
	
	\section{Mathematical Preliminaries} \label{sec:mathprelim}
	
	We now define some mathematical terms that are used in later sections.  Given two Banach spaces $P$ and $Q$, $\mathcal{L}(P,Q)$ is the space of bounded operators from $P$ to $Q$. For a finite collection of Banach spaces, $\lbrace Z_1,...,Z_M \rbrace$, we equip the Cartesian product of the spaces, $Z_1 \times ... \times Z_M$, with the norm  $\big( \sum_{\alpha=1}^M \| \cdot \|^2_\alpha \big)^{1/2}$. We define $L^2(\Omega)$ as the Hilbert space of square integrable functions over $\Omega$, where $\Omega \subset \mathbb{R}^n$ is a bounded open subset of a Euclidean domain of dimension  $n$. The Hilbertian structure of $L^2(\Omega)$ is induced by the standard inner product, $\langle \cdot , \cdot \rangle_{L^2(\Omega)} :  L^2(\Omega) \times L^2(\Omega) \rightarrow \mathbb{R}$, given by
	$
	\langle f, g\rangle_{L^2(\Omega)} = \int_{\Omega}  f(\mathbf{x})g(\mathbf{x})d\mathbf{x}
	$
	for each $f,g \in L^2(\Omega)$. The norm $\|\cdot\|_{L^2(\Omega)}$ on the space $L^2(\Omega)$ is defined as
	$
	\|f\|_{L^2(\Omega)} = \langle f, f\rangle^{1/2}_{L^2(\Omega)}
	$
	for all $f \in L^2(\Omega)$. For a positive integer $n \in \mathbb{Z}^+$, $P^n$ refers to the Cartesian product of $n$ copies of the space $P$; that is, $P^n = \prod_{i=1}^n P$.  We use $P^*$ to denote the space of linear continuous functionals on the Banach space $P$. The bilinear form that induces the duality between $P$ and $P^*$ is given by $\langle \cdot, \cdot \rangle_{P^*,P}:P^* \times P \rightarrow \mathbb{R}$, where $\langle x , y \rangle_{P^*,P} = x(y)$ for each $x \in P^*$ and each $y \in P$. 
	
	We define the Sobolev space $H^1(\Omega) = \big \lbrace z \in L^2(\Omega): \frac{\partial z}{\partial x_1}, \frac{\partial z}{\partial x_2} \in L^2(\Omega) \big \rbrace$. Here, the spatial derivatives are to be understood as weak derivatives defined in the distributional sense. See  \cite{evans1998partial} for this notion of a derivative of a function that is not necessarily differentiable in the classical sense. We equip the space with the usual Sobolev norm 
	$
	\|y\|_{H^1(\Omega)} = \Big( \|y\|^2_{L^2(\Omega)} + \sum_{i=1}^2 \left\| \frac{\partial y}{\partial x_i} \right\|^2_{L^2(\Omega)}\Big)^{1/2}. \nonumber
	$
	The dual space of $H^1(\Omega)$, denoted by $H^1(\Omega)^*$, is the space of bounded linear functionals on $H^1(\Omega)$ through the inner product of $L^2(\Omega)$. Defining $V=H^1(\Omega)$ and $X = V \times L^2(\Omega)^{2}$, it follows that $X^* := V^{*} \times L^2(\Omega)^{2}$. 
	
	We say that a sequence $f_n$ is weakly converging to $f$ in $X$, written as $f_n\rightharpoonup f$, if $
	\lim_{n \rightarrow \infty} \langle \phi,f_n \rangle_{X^*,X} =  \langle \phi,f \rangle_{X^*,X} \nonumber
	$ for all $\phi$ in $X^*$. For norm convergence of a sequence $f_n$ to $f$, we write $f_n \rightarrow f$. 
	
	The space $L^2(0,T;X)$ consists of all strongly measurable functions $u:[0,T] \rightarrow X$ for which  $ \|u\|_{L^2(0,T;X)} := \big ( \int^T_0 \|u(t)\|_X^2 \big )^{1/2} ~<~ \infty $. The space $C([0,T];X)$ consists of all continuous functions $u:[0,T] \rightarrow X$ for which $ \| u\|_{C([0,T];X)} = \max_{0 \leq t \leq T} \|u(t)\|_X < \infty$. 
	
	We will also use the {\it trace operator}, $\tau : H^1(\Omega) \rightarrow L^2(\partial \Omega)$, defined as the unique linear bounded operator that satisfies $(\tau u)(\mathbf{x}) = u(\mathbf{x})$ for every $\mathbf{x} \in \partial \Omega$ whenever $u$ is in $C^1(\bar{\Omega})$ (that is, $u$ is at least once differentiable in the classical sense). Informally speaking, the trace operator gives meaning to the evaluation of an element of $H^1(\Omega)$ on $\partial \Omega$, which is a set of measure zero.
	
	\section{Optimization Problem for the Mapping Assignment} \label{sec:optMap}
	
	
	This section formulates the {\it mapping assignment} as an optimization problem that identifies the unknown coefficient $H_{\Gamma}$ in the  {\it mapping} macroscopic model \eqref{eq:MapBCS}. 
	This coefficient is reconstructed from the cumulative robot sensor data $\hat{g}(t) $, the total number of observations made over the entire domain by all the robots up until time $t$. We recall that the robots record these observations stochastically according to the chemical reaction \eqref{eq:CRN_map} and that the agents' spatial states evolve according to equation \eqref{eq:RWPTpos}. Here, the velocity input $\mathbf{v}(t)$ is predetermined by the user and should be chosen such that the agents approximately cover the entire domain. An example of such a  choice for $\mathbf{v}(t)$ is one that would drive the robots along a lawnmower path, as shown in Fig. \ref{fig:AM} in Section \ref{sec:simulat}.
	
	In the mapping macroscopic model \eqref{eq:MapBCS}, the cumulative robot sensor data is interpreted as a continuous quantity given by an integral over the domain, $g(t) = \int_{\Omega}y_2({\bf x},t)d{\bf x}$.
	We address the existence and uniqueness of the solutions of the macroscopic model \eqref{eq:MapBCS} in Section \ref{sec:optCoverage}, since this model is a special case of the {\it coverage} macroscopic model \eqref{eq:Macro1} for which $k_f = 0$, $k(t) = k_o$, and the {\it stationary} robot state is absent.
	
	
	We define $\hat{H}$ as an estimate of $H_{\Gamma}$. Due to the one-sided coupling between the state variables $y_1$ and $y_2$ in the PDE model \eqref{eq:MapBCS}, the variable $y_1$ is not affected by the coefficient $H_{\Gamma}$. Hence, assuming that the PDE is well-posed and that a unique solution exists in $H^1(\Omega) \times L^2(\Omega)$, we can pose the mapping problem follows: 
	
	\vspace{2mm}
	\begin{problem} [{Mapping Problem}]  \label{mappingtask_meanfield}
		Given an initial condition $y_0({\bf x})$ of PDE model \eqref{eq:MapBCS} and the corresponding solution $y_1({\bf x}, t)$, find an estimate $\hat{H} \in S_{ad}$, where 
		\begin{equation}
		S_{ad} ~:=~ \big \lbrace u\in L^2(\Omega);  ~0 \leq u(\mathbf{x}) \leq 1 \hspace{2mm} a.e. \hspace{2mm} \mathbf{x} \in\Omega \big \rbrace,\nonumber
		\end{equation} 
		that satisfies the following equation:
		\begin{equation}
		\label{eq:Koperator}
		(K\hat{H})(t) ~:=~ \int_{\Omega} k_o \hat{H}(\mathbf{x}) y_{1}(\mathbf{x},t)d\mathbf{x} ~=~ g(t).
		\end{equation}
	\end{problem} 
	Ideally, the estimate $\hat{H}$ would be constrained to be a function that takes values of $0$ or $1$ over its domain. However, this constraint would make the resulting inverse problem, i.e. solving equation \eqref{eq:Koperator} for $\hat{H}$ in Problem \ref{mappingtask_meanfield}, non-convex when considered as an optimization problem. To make the optimization problem tractable, we relax the range of $\hat{H}$ and consider it as an essentially bounded element of $L^2(\Omega)$.  
	Generally, equations of the form \eqref{eq:Koperator} need not have unique solutions unless some special conditions on $k_o y_1(\mathbf{x},t)$, the kernel of the integral operator $K:L^2 (\Omega) \rightarrow L^2(0,T)$, can be guaranteed. To resolve the ill-posedness of the inverse problem, it can be alternatively posed as an optimization problem with a functional $J(\hat{H})$ that is convex, but not necessarily strictly convex:
	\begin{equation}
	\min_{\hat{H} \in S_{ad}} \hspace{2mm} J(\hat{H})=\|K\hat{H}-g\|^2_{L^2(0,T)}
	\label{eq:Mapprob}
	\end{equation}
	To ensure a unique solution for $\hat{H}$, the optimization problem can be recast with a strictly convex functional:
	\begin{equation}
	\min_{\hat{H} \in S_{ad}} J_{\lambda}(\hat{H})=\frac{1}{2}\|K\hat{H}-g\|^2_{L^2(0,T)}+\frac{\lambda}{2} \|\hat{H}\|^2_{L^2(\Omega)},
	\label{eq:RegMap}
	\end{equation}
	where $\lambda >0$ is the regularization parameter that is often used in the so-called ``Tikohnov regularization'' of inverse  problems \cite{kirsch2011introduction}. Since $S_{ad}$ is convex in $L^2(\Omega)$, the existence and uniqueness of regularized problems of the form \eqref{eq:RegMap} can be guaranteed and has been been well-studied in the theory of inverse problems \cite{kirsch2011introduction}. Hence, we have the following result in \autoref{th:uniquesol}.  The assumption in the theorem that $y_1 \in C([0,T];L^2(\Omega))$ is justified in Section \ref{sec:Engyo}.
	
	\vspace{2mm}
	
	\begin{theorem} \label{th:uniquesol}
		A unique solution to the regularized inverse problem \eqref{eq:RegMap} exists for each $\lambda > 0$, under the assumptions that $y_1 \in C([0,T];L^2(\Omega))$ and $k_o>0$. 
	\end{theorem}
	\begin{proof}
		For $y_1 \in C([0,T];L^2(\Omega))$ and $k_o>0$, the operator $K$ in equation \eqref{eq:Koperator} is a bounded operator from $L^2(\Omega)$ to $L^2(0,T)$. Thus, the result follows from \cite{kirsch2011introduction}[Theorem 2.11].
	\end{proof}
	
	\vspace{2mm}
	
	We use a gradient descent method to solve the optimization problem \eqref{eq:RegMap}.  To apply this method, we need to characterize the derivative of the objective functional $J_{\lambda}(\hat{H})$. This functional is differentiable in the Fr\'echet sense. Since $K \in \mathcal{L}(L^2(\Omega),L^2(0,T))$, the derivative of $K$ is itself. Then by the chain rule of differentiation, the Fr\'echet derivative of $J_{\lambda}(\hat{H})$, denoted by $J_{\lambda}'(\hat{H})$, is defined as 
	\begin{equation}
	\langle J_{\lambda}'(\hat{H}),s \rangle_{L^2(\Omega)} = \langle K\hat{H}-g,Ks \rangle  _{L^2(0,T)} + \lambda  \langle \hat{H},s \rangle_{L^2(\Omega)}.
	\end{equation}
	Using Riesz representation \cite{evans1998partial}[Appendix D, Theorem 2], we can obtain an explicit representation of the gradient of $J_{\lambda}(\hat{H})$,
	\begin{equation}
	\nabla J_{\lambda}(\hat{H}) ~=~ K^*(K\hat{H}-g)+\lambda \hat{H} ~\in~ L^2(\Omega),
	\label{eq:objmap}
	\end{equation}
	where $K^* \in \mathcal{L}(L^2(0,T),L^2(\Omega))$ is defined as
	\begin{equation}
	(K^* G)(\mathbf{x}) =  \int_0^T k_o G(t)y_1(\mathbf{x},t)dt \hspace{5mm} \forall G \in L^2(0,T).
	\end{equation} 
	To verify that this characterization of $K^*$ is correct, it is straightforward to check that
	\begin{align}
	\langle K\hat{H},G \rangle_{L^2(0,T)} - \langle \hat{H},K^*G \rangle_{L^2(\Omega)} ~=~ 0 \hspace{5mm} \\ \nonumber 
	\forall \hat{H} \in L^2(\Omega), \hspace{2mm} \forall G \in L^2(0,T).
	\end{align}

	\section{Optimal Control Problem for the Coverage Assignment}  \label{sec:optCoverage}
	
	In Section \ref{sec:formoptcontrol}, we formulate the {\it coverage assignment} as an optimal control problem that is used to compute the parameters $\mathbf{v}(t)$ and $k(t)$ in the {\it coverage} macroscopic model \eqref{eq:Macro1}. We analyze the existence and uniqueness of solutions to the PDE model \eqref{eq:Macro1} in Section \ref{sec:Engyo}, the well-posedness of the optimal control in Section \ref{sec:ExOpt}, and the differentiability of the objective functional in Section \ref{sec:diff}.
	%
	
	\subsection{Formulation of the Optimal Control Problem} \label{sec:formoptcontrol}
	
	We first express the PDE model \eqref{eq:Macro1} as a bilinear control system with control inputs $\mathbf{v}(t)$ and $k(t)$.  Toward this end, we supply the following definitions.  Recall from Section \ref{sec:mathprelim} that $V = H^1(\Omega)$,  $X = V \times L^2(\Omega)^{2}$, and $X^* := V^{*} \times L^2(\Omega)^{2}$, and that $\mathcal{L}(X,L^2(\Omega)^3)$ denotes the space of linear bounded operators from $X$ to $L^2(\Omega)^3$. We define the operators $A$, $B_i \in \mathcal{L}(X,L^2(\Omega)^3)$, $i=1,2,3$, as follows:
	\begin{eqnarray}
	&& A =
	\begin{bmatrix}
	D\nabla^2 & k_f & 0\\  
	0 & -k_f & 0\\
	0 & 0 & 0
	\end{bmatrix}
	\hspace{2mm} B_1 = 
	\begin{bmatrix}
	-\frac{\partial}{\partial x_1}   & 0 & 0\\  
	0 & 0 & 0\\
	0 & 0 & 0
	\end{bmatrix}
	\nonumber \\
	&& B_2 = 
	\begin{bmatrix}
	-\frac{\partial}{\partial x_2}   & 0 & 0\\  
	0 & 0 & 0\\
	0 & 0 & 0
	\end{bmatrix}
	\hspace{2mm} B_3 = 
	\begin{bmatrix}
	-H_{\Gamma}   & 0 & 0\\  
	H_{\Gamma} & 0 & 0\\
	H_{\Gamma} & 0 & 0
	\end{bmatrix}.\nonumber 
	\end{eqnarray}
	We also define the spaces $F = L^2(0,T;L^2(\Omega)^{3})$, $G= L^2(0,T;L^2(\partial \Omega))$, $Y = L^2(0,T;X)$, and $Y^* =   L^2(0,T;X^*)$. These definitions will be used to  establish the well-posedeness of the PDE \eqref{eq:Macro1} using the classical theory of {\it weak solutions} of linear PDEs \cite{evans1998partial}. For $\mathbf{f} \in F$, $g \in G$, and $\mathbf{y}_0 \in L^2(\Omega)^{3}$, let  $\mathbf{y} \in Y$ be a function with time derivative $\partial \mathbf{y} / \partial t \in Y^*$.  Denoting the components of the robots' velocity field by $\mathbf{v}(t) = [v_x(t)~v_y(t)]^T$, we define the control inputs $u_1(t) = v_x(t)$, $u_2(t) = v_y(t)$, and $u_3(t) = k(t)$. We can now write the PDE \eqref{eq:Macro1} with boundary condition \eqref{eq:NFBC} and initial condition \eqref{eq:covIC} as a bilinear control system in the following form: 
	\begin{eqnarray}
	&& \hspace{5mm} \frac{\partial \mathbf{y}}{\partial t} = A\mathbf{y} + \sum_{i=1}^{3} u_{i}B_{i}\mathbf{y} + \mathbf{f} \hspace{5mm} in \hspace{2mm} Q, \nonumber \\
	&& \hspace{5mm} \mathbf{n} \cdot (\nabla y_1 - [u_1 ~ u_2]^T y_1)=g  \hspace{5mm} on \hspace{2mm} \Sigma, \nonumber \\
	&&  \hspace{5mm} \mathbf{y}(0)=\mathbf{y}_0 \hspace{5mm} on \hspace{2mm} \Omega, \label{eq:BCS1}
	\end{eqnarray}
	where $Q$ and $\Sigma$ are as defined in Section \ref{sec:macrosec}. 
	We note that the solution of the PDE model in equations \eqref{eq:Macro1}, \eqref{eq:NFBC}, \eqref{eq:covIC} is the solution of the system \eqref{eq:BCS1} with $\mathbf{f}=\mathbf{0}$ and $g=0$.  We consider the more general form \eqref{eq:BCS1} for the purpose of analyzing the differentiability properties of the  control-to-state map (defined in Section \ref{sec:diff}) and the objective functional in the optimal control problem.
	
	A function $\mathbf{y}$ is a weak solution of system \eqref{eq:BCS1} if 
	\begin{align}
	\label{eq:WeakBCS}
	\langle \frac{\partial \mathbf{y}}{\partial t},\boldsymbol{\phi} \rangle_{Y^*,Y} = ~& \int_0^T\langle A_g(t)\mathbf{y}(t),\boldsymbol{\phi}(t) \rangle_{X^*,X}dt \\ \nonumber
	& + \sum_{i=1}^{3}  \langle u_{i}B_{i}\mathbf{y},\boldsymbol{\phi} \rangle_{F} + \langle \mathbf{f},\boldsymbol{\phi} \rangle_{F}
	\end{align} 
	for all $\boldsymbol{\phi} \in L^2(0,T;X)$. Here, $A_g(t) : X \rightarrow X^*$ is the variational form of the operator $A$ for each $t \in (0,T)$. 
	The boundary condition \eqref{eq:NFBC} is equipped with $A_g$ in the variational formulation using Green's theorem as, 
	\begin{equation}
	\label{eq:weaklap}
	A_g(t) =
	\begin{bmatrix}
	M_g(t) & k_f & 0\\  
	0 & -k_f & 0\\
	0 & 0 & 0
	\end{bmatrix},
	\end{equation} 
	where $M_g(t):V \rightarrow V^*$ for each $t \in (0,T)$ is the Laplacian in the variational formulation and is defined as
	\begin{equation}
	\begin{split}
	\left\langle M_g(t) y,\phi \right\rangle_{V^*,V} ~=~ & - \left\langle D \nabla y,\nabla \phi \right\rangle_{L^2(\Omega)} \\ 
	& \hspace{-5mm} + \int_{\partial \Omega} (g(t)+\mathbf{n} \cdot [u_1(t) ~ u_2(t)]^Ty) \phi d\mathbf{x}
	\label{eq:Mdef}
	\end{split}
	\end{equation} 
	for all $y \in V$ and $\phi \in V^*$.
	
	The coverage problem can now be framed as follows.
	\vspace{2mm}
	\begin{problem} [\textbf{Coverage Problem}]  \label{coveragetask_meanfield}
		Define a  target spatial distribution of robot coverage activity, $\mathbf{y}_{\Omega} \in L^2(\Omega)^3$, to be achieved by time $T$, and a set of admissible control inputs,
		\begin{equation}
		\begin{split}
		U_{ad} = \lbrace \mathbf{u} \in L^2(0,T)^{3}; & ~~u^{min}_i \leq u_i(t) \leq u^{max}_i, \hspace{2mm} \\
		& ~~ i=1,2,3, ~~a.e. \hspace{2mm} t \in (0,T) \rbrace. \nonumber
		\end{split}
		\end{equation}
		Let $\bar{Y} = C([0,T],L^2(\Omega)^{3})$, and define a function $W \in \mathcal{L}(L^2(\Omega)^{3})$ that weights the relative significance of minimizing the distances between different state variables and their target distributions.  Then, given an initial condition $y_0({\bf x})$, solve the optimal control problem
		\begin{equation}
		\begin{split}
		\min\limits_{(\mathbf{y},\mathbf{u})\in \bar{Y} \times U_{ad}} J(\mathbf{y},\mathbf{u}) ~=~ & \frac{1}{2} \|W\mathbf{y}(\cdot,T)-\mathbf{y}_{\Omega}\|_{L^2(\Omega)^{3}}^2 \\
		& + \frac{\lambda}{2}\|\mathbf{u}\|_{L^2(0,T)^{m}}^2  \label{eq:OptControlProblem}
		\end{split}
		\end{equation}
		subject to equation \eqref{eq:BCS1} with $\mathbf{f}=\mathbf{0}$ and $g=0$.   
	\end{problem} 
	
	Note that, due to the essential bounds on $\mathbf{u}$, we have that $\mathbf{u} \in L^{\infty}(0,T)^{3}$. Here, we only consider essentially bounded control inputs since we only prove existence of solutions for control variables in $L^{\infty}(0,T)^{3}$. The existence and uniqueness of solutions for control variables in $L^p(0,T)^{3}$, $p < \infty$, is a nontrivial issue and is beyond the scope of this work.

	\subsection{Energy Estimates} \label{sec:Engyo}
	
	Energy estimates refer to bounds on the solutions of a PDE system with respect to certain parameters of interest, such as the initial condition, coefficients, or boundary parameters. Whereas in the theory of weak solutions of PDEs \cite{evans1998partial}, these energy estimates are used to show the existence of solutions, in optimal control analysis they are used to study the differentiability properties of the control-to-state map. In this section, we derive energy estimates for the solutions of the PDE model \eqref{eq:WeakBCS} and use them to establish existence and uniqueness of solutions to the model, 
	a result (Lemma \ref{ExistLemma}) that we previously stated without proof in our work \cite{elamvazhuthi2015optimal}. More importantly, the estimates are then used to prove existence of solutions to the optimal control problem. The differentiability of the control-to-state map also follows from these energy estimates. 
	
	
	\vspace{2mm}
	
	\begin{lemma} \label{lemma:energyest}
		Let $\mathbf{b} \in \mathbb{R}^2$ be such that  $\max  \limits_{i=1,2} |b_i| \leq \max \limits_{i=1,2} |u^{max}_i|+\min \limits_{i=1,2} |u^{max}_i|$. Suppose that $\tilde{g} \in L^2(\partial \Omega)$. Define $M:V \rightarrow V^*$ as 
		\begin{equation}
		\left\langle My,\phi \right\rangle_{V^*,V} =  \left\langle D \nabla y,\nabla \phi \right\rangle_{L^2(\Omega)} - \int_{\partial \Omega} (\tilde{g}+\mathbf{n} \cdot \mathbf{b}y) \phi d\mathbf{x}.
		\label{eq:adntM}
		\end{equation} 
		Then we have the following energy estimate for all $y \in V$:
		\begin{equation}
		\tilde{\beta} \|y\|^2_V ~\leq~ \left\langle My,y \right\rangle_{V^*,V} + \tilde{\alpha} (\|\tilde{g}\|^2_{L^2(\partial \Omega)} + \|y\|^2_{L^2(\Omega)}),
		\end{equation} 
		which holds for some $\tilde{\alpha}, \tilde{\beta}>0 $ that depend only on $\Omega$, $\max \limits_{i=1,2} |u^{min}_i|$, and $\max \limits_{i=1,2} |u^{max}_i|$.
	\end{lemma}
	
	\begin{proof}
		Setting $\phi = y$ in definition \eqref{eq:Mdef}, we obtain the inequalities 
		\begin{equation}
		D \int_{\Omega} | \nabla y|^2 d\mathbf{x} ~ \leq ~ \left\langle My,y \right\rangle_{V^*,V} + \int_{\partial \Omega} (\tilde{g} + \mathbf{n} \cdot \mathbf{b} y)y d\mathbf{x}  \nonumber
		\end{equation}
		\begin{equation}
		\begin{split}
		\hspace{10mm}   \leq \left\langle My,y \right\rangle_{V^*,V} &+ \int_{\partial \Omega} |\tilde{g}||y| d\mathbf{x} 
		+ \| \mathbf{b} \| \int_{\partial \Omega} |y|^2 d\mathbf{x}, \nonumber
		\end{split}
		\end{equation}
		from which we can conclude that
		\begin{equation}
		\begin{split}
		D \|y\|^2_{V}  \leq \left\langle My,y \right\rangle_{V^*,V} &+ \frac{1}{2} \|\tilde{g}\|^2_{L^2(\partial \Omega)} + \frac{1}{2} \|\tau  y\|^2_{L^2(\partial \Omega)} \\
		&+ \|\mathbf{b}\| \|\tau y\|^2_{L^2(\partial \Omega)} + D \|y\|^2_{L^2(\Omega)}, 
		\label{eq:lemeq1}
		\end{split}
		\end{equation}
		where $\tau$ is the trace operator defined in Section \ref{sec:mathprelim}. Then it follows from a modified form of the {\it Trace Theorem} for domains with Lipschitz boundaries \cite{grisvard2011elliptic}[Theorem 1.5.1.10] that there exists a constant $K$, independent of $\epsilon ~\in (0,1)$ and $y \in V$, such that the following inequality holds for every $\epsilon \in (0,1)$:
		\begin{align*}
		\|\tau y\|^2_{L^2(\partial \Omega)} \leq K  \epsilon^{1/2}\|y\|^2_{V}+K (\epsilon^{-1/2}-\epsilon^{+1/2})\|y\|^2_{L^2(\Omega)}.
		\end{align*}
		Thus, for each $\epsilon \in (0,1)$, the inequality \eqref{eq:lemeq1} becomes 
		\begin{equation}
		\begin{split}
		D \|y\|^2_{V} \leq \left\langle My,y \right\rangle_{V^*,V} &+ \frac{1}{2} \|\tilde{g}\|^2_{L^2(\partial \Omega)} + \tilde{b}_{\epsilon}   \|y\|^2_{V} \\
		& + C_{\epsilon} \|y\|^2_{L^2(\Omega)} , 
		\end{split}
		\end{equation}
		where $\tilde{b}_{\epsilon} = \frac{K}{2}  \epsilon^{1/2}+rK  \epsilon^{1/2}$, $C_{\epsilon} = D+\frac{K}{2}(\epsilon^{-1/2}-\epsilon^{+1/2})+rK(\epsilon^{-1/2}-\epsilon^{+1/2})$, and $r = \max \limits_{i=1,2} |u^{max}_i|+\min \limits_{i=1,2} |u^{max}_i|$. Note that we can choose an $\epsilon >0$ arbitrarily small, independent of $y \in V$, such that $D > \tilde{b}_{\epsilon}$. We fix such an $\epsilon>0$. Then we set $\tilde{\beta}=D - \tilde{b}_{\epsilon}$ and $\tilde{\alpha} = \max \lbrace \frac{1}{2},C_{\epsilon} \rbrace$. This concludes the proof.
	\end{proof}
	
	
	\vspace{2mm}
	
	\begin{corollary} \label{Corren}
		Define $A_g(t): X \rightarrow X^*$ as in equation \eqref{eq:weaklap}. Then, for almost every $t \in (0,T)$, we have the following energy estimate:
		\begin{eqnarray}
		\beta \|\mathbf{y}\|^2_X  ~~\leq ~~& -& \left\langle A_g(t)\mathbf{y},\mathbf{y} \right\rangle_{X^*,X} \nonumber \\
		~~& + & ~\alpha (\|g(t)\|^2_{L^2(\partial \Omega)} + \|\mathbf{y}\|^2_{L^2(\Omega)^{3}}),
		\label{eq:energy}
		\end{eqnarray}
		which holds for some $\alpha, \beta>0 $ that depend only on $\Omega$, $\max \limits_{i=1,2} |u^{min}_i|$, and $\max \limits_{i=1,2} |u^{max}_i|$.
	\end{corollary}
	
	\begin{proof}
		
		Comparing equations \eqref{eq:Mdef} and \eqref{eq:adntM}, we observe that $M_g(t) = -M$, $\mathbf{b}=[u_1(t)~u_2(t)]^T$, and $\tilde{g} = g(t)$ for almost every $t \in (0,T)$. Then, from the definition \eqref{eq:weaklap} of the operator $A_g(t)$, we have that
		\begin{align*}
		-\left\langle A_g(t)\mathbf{y},\mathbf{y} \right\rangle_{X^*,X} = \left\langle My_1,y_1 \right\rangle_{V^*,V}&- k_f \left\langle y_1,y_2 \right\rangle_{L^2(\Omega)} \\ \nonumber 
		&+k_f \|y_2\|^2_{L^2(\Omega)}.
		\end{align*}
		Using Cauchy's inequality and Young's inequality \cite{evans1998partial},
		\begin{align*}
		\left\langle M\mathbf{y},\mathbf{y} \right\rangle_{X^*,X} ~\leq~ & - \left\langle A_g(t)y_1,y_1 \right\rangle_{V^*,V}+ \frac{k_f}{2}\|y_1\|^2_2 \\ \nonumber  &+ \frac{k_f}{2}\|y_2\|^2_2   
		+k_f \|y_2\|^2_{L^2(\Omega)}.
		\end{align*}
		We set $\beta = \tilde{\beta}$ and $\alpha=\max \lbrace \tilde{\alpha},\frac{3}{2}k_f \rbrace$, where $ \tilde{\alpha}$ and $\tilde{\beta}$ are defined as in the proof of Lemma \ref{lemma:energyest}. This concludes the proof.
	\end{proof}
	
	\vspace{2mm}
	
	
	\begin{lemma}
		\label{ExistLemma}
		~~~Given $\mathbf{f} \in L^2(0,T;L^2(\Omega)^{3})$, $ g \in$ $L^2(0,T;L^2(\partial \Omega))$, and an initial condition $\mathbf{y}_0 \in L^2(\Omega)^{3}$, there exists a unique solution $\mathbf{y} \in C([0,T];L^2(\Omega)^{3})$ to system \eqref{eq:BCS1}. We have the following energy estimate for this solution:  
		\begin{align}
		& \hspace{-3mm} \|\mathbf{y}\|_{C([0,T];L^2(\Omega)^{3})} +\|\mathbf{y}\|_{L^2(0,T;X)} +\|\partial \mathbf{y} / \partial t \|_{L^2(0,T;X^*)} \nonumber \\ 
		\leq & \hspace{1mm}  K  (\|\mathbf{y}_0\|_{L^2(\Omega)^{3}}+ \|\mathbf{f}\|_{L^2(0,T;L^2(\Omega)^3)} + \|g\|_{L^2(0,T;L^2(\partial \Omega)}),
		\label{eq:BCSenergy}
		\end{align} 
		where $K>0$ depends only on $\Omega$, $\max \limits_{1 \leq i \leq 3} |u^{max}_i|$, and $\max \limits_{1 \leq i \leq 3} |u^{min}_i|$. 
	\end{lemma}
	\begin{proof}
		We first determine a bound on the sum $ \|\mathbf{y}\|_{C([0,T];L^2(\Omega)^{3})} +\|\mathbf{y}\|_{L^2(0,T;X)}$ in the energy estimate \eqref{eq:BCSenergy}. Let $\mathbf{y}$ be a weak solution of system \eqref{eq:WeakBCS}, and set  $\mathbf{\phi} = \mathbf{y}$ in \eqref{eq:WeakBCS}.
		Then 
		\begin{align*}
		\left\langle \frac{\partial \mathbf{y}}{\partial t}(t),\mathbf{y}(t) \right\rangle_{X,X^*} - \left\langle A_g(t) \mathbf{y}(t),\mathbf{y}(t) \right\rangle _{X,X^*} \\ \nonumber
		~=~ \sum_{i=1}^3 \left\langle u_{i}(t)B_{i}\mathbf{y}(t),\mathbf{y}(t) \right\rangle _{L^2(\Omega)^{3}}  \\ \nonumber 
		+ ~\left\langle \mathbf{f}(t),\mathbf{y}(t) \right\rangle  _{L^2(\Omega)^{3}}.
		\end{align*}
		Combining this equality with the inequality \eqref{eq:energy}, and observing that $\left\langle \frac{\partial \mathbf{y}}{\partial t}(t),\mathbf{y}(t) \right\rangle_{X,X^*} = \frac{1}{2} \frac{d}{dt} \|\mathbf{y}(t)\|^2_{L^2(\Omega)^{3}}$, we obtain the inequality
		\begin{align*}
		& \frac{d}{dt} \|\mathbf{y}(t)\|^2_{L^2(\Omega)^{3}} + 2 \beta  \|\mathbf{y}(t)\|^2_X \nonumber \\
		& \leq~ 2 \sum_{i=1}^3 \|u_{i}(t)\|_{L^{ \infty }(0,T)} \left\langle B_{i}\mathbf{y}(t),\mathbf{y}(t) \right\rangle _{L^2(\Omega)^{3}} \\ \nonumber
		& + 2 \left\langle \mathbf{f}(t),\mathbf{y}(t) \right\rangle  _{L^2(\Omega)^{3}} + 2 \alpha (\|g(t)\|^2_{L^2(\partial \Omega)} + \|y_1(t)\|^2_{L^2(\Omega)}). 
		\end{align*}
		Using Cauchy's inequality and Young's inequality \cite{evans1998partial}, we derive the following inequality:
		\begin{align}
		\frac{d}{dt} & \|\mathbf{y}(t)\|^2_{L^2(\Omega)^{3}} + 2 \beta \|\mathbf{y}(t)\|^2_X \nonumber \\
		& ~~\leq ~ 2\sum_{i=1}^p \|u_{i}\|_{L^{ \infty }(0,T)} (\| B_{i}\mathbf{y}(t)\|_{L^2(\Omega)^{3}} \|\mathbf{y}(t) \|_{L^2(\Omega)^{3}}) \nonumber \\ 
		& ~~~~~ + ~(\|\mathbf{f}(t)\|^2_{L^2(\Omega)^{3}} + \|\mathbf{y}(t)\|^2_{L^2(\Omega)^{3}})\nonumber \\ 
		&~~~~~  + ~2 \alpha  (\|g(t)\|^2_{L^2(\partial \Omega)} + \|\mathbf{y}(t)\|^2_{L^2(\Omega)^3}).  \label{eq:enest_ineq1}
		\end{align} 
		Let $E = r \times (\max \limits_{1 \leq i \leq 3} |u^{max}_i| +  \max \limits_{1 \leq i \leq 3} |u^{min}_i|$), where $r$ is the maximum of the operator norms of the operators $B_i$, $i = 1,2,3$. Then we have a bound on the first term on the right side of inequality \eqref{eq:enest_ineq1} as follows:
		\begin{align}
		\hspace{-5mm} 2\sum_{i=1}^3 \|u_{i}\|_{L^{ \infty }(0,T)} (\| B_{i}\mathbf{y}(t)\|_{L^2(\Omega)^{3}} \|\mathbf{y}(t) \|_{L^2(\Omega)^{3}}) \nonumber \\
		~ \leq ~ 2 E (\|\mathbf{y}(t)\|_{X} \|\mathbf{y}(t) \|_{L^2(\Omega)^{3}}). \label{eq:enest_ineq2}
		\end{align}
		Applying Young's inequality, we can bound the term on the right side of inequality \eqref{eq:enest_ineq2}:
		\begin{equation}
		2 E (\|\mathbf{y}(t)\|_{X} \|\mathbf{y}(t) \|_{L^2(\Omega)^{3}}) 
		\leq \beta \|\mathbf{y}(t)\|^2_{X} + \frac{E^2}{\beta}\|\mathbf{y}(t) \|^2_{L^2(\Omega)^{3}}.
		\end{equation}
		Combining these upper bounds with the inequality \eqref{eq:enest_ineq1}, we find that:
		\begin{align*}
		\frac{d}{dt} & \|\mathbf{y}(t)\|^2_{L^2(\Omega)^{3}} + \beta \|\mathbf{y}(t)\|^2_X \nonumber \\
		& \leq ~\frac{E^2}{\beta}\|\mathbf{y}(t) \|^2_{L^2(\Omega)^{3}} + (\|\mathbf{f}\|^2_{L^2(\Omega)^{3}} + \|\mathbf{y}(t)\|^2_{L^2(\Omega)^{3}}) \nonumber \\
		& ~~~~~ + 2\alpha (\|g(t)\|^2_{L^2(\partial \Omega)} + \|\mathbf{y}(t)\|^2_{L^2(\Omega)^{3}}),
		\end{align*}
		and therefore
		\begin{align}
		\frac{d}{dt} & \|\mathbf{y}(t)\|^2_{L^2(\Omega)^{3}} + \beta   \|\mathbf{y}(t)\|^2_X \nonumber \\
		& ~~~ \leq ~ C(\|\mathbf{y}(t) \|^2_{L^2(\Omega)^{3}} + \|\mathbf{f}(t)\|^2_{L^2(\Omega)^{3}} + \|g(t)\|^2_{L^2(\partial \Omega)}) \label{eq:en11}
		\end{align} 
		for $C =  \frac{E^2}{\beta} + 1 + 2\alpha$.  This inequality implies that
		\begin{align*}
		\frac{d}{dt} \|\mathbf{y}(t)\|^2_{L^2(\Omega)^{3}} \leq  C(\|\mathbf{y}(t) \|^2_{L^2(\Omega)^{3}} + \|\mathbf{f}(t)\|^2_{L^2(\Omega)^{3}} +\|g(t)\|^2_{L^2(\partial \Omega)}).
		\end{align*}
		Setting $\eta(t) = \|\mathbf{y}(t)\|^2_{L^2(\Omega)^{3}}$ and $\psi(t) = C(\|\mathbf{f}(t)\|^2_{L^2(\Omega)^{3}} + \|g(t)\|^2_{L^2(\partial \Omega)})$, and applying Gronwall's lemma \cite{evans1998partial} with these functions, we get:
		\begin{align*}
		\max_{0 \leq t \leq T} \|\mathbf{y}(t)\|^2_{L^2(\Omega)^{3}} \leq & ~~e^{CT} (\|\mathbf{y}_0 \|^2_{L^2(\Omega)^{3}}  \nonumber \\
		& + \|\mathbf{f}\|^2_{L^2(0,T;L^2(\Omega)^{3})} \nonumber \\ 
		& +\|g\|^2_{L^2(0,T);L^2(\partial \Omega)} ). 
		\end{align*}
		Then integrating both sides of inequality \eqref{eq:en11} over the time interval $(0,T)$, and using the inequality above, we obtain a bound on the sum of the first two terms in the energy estimate \eqref{eq:BCSenergy}: 
		\begin{align}
		\label{eq:prevest} 
		\max_{0 \leq t \leq T} \|\mathbf{y}(t)\|^2_{L^2(\Omega)^{3}} +  \|\mathbf{y}\|^2_{L^2(0,T;X)} & \\ \nonumber 
		\leq C'(\|\mathbf{y}_0 \|^2_{L^2(\Omega)^{3}} & + \|\mathbf{f}\|^2_{L^2(0,T;L^2(\Omega)^{3})}  \\ \nonumber
		& +\|g\|^2_{L^2(0,T);L^2(\partial \Omega)}),
		\end{align}
		where $C' = {\rm max}\lbrace e^{CT},\frac{C}{\beta}\rbrace$.
		
		Next, we derive a bound on the third term in estimate \eqref{eq:BCSenergy},
		$\|\partial \mathbf{y} / \partial t \|_{L^2(0,T;X^*)}$. We know that $B_i$ is a bounded operator from $X$ to $L^2(\Omega)^3$ for each $i \in \lbrace 1, 2, 3 \rbrace$. Therefore, $\|B_i \mathbf{q}\|_{L^2(\Omega)^2} \leq c_i \|\mathbf{q}\|_X$ for all $\mathbf{q} \in X$ and for some positive constants $c_i$. Similarly, $A_g \in \mathcal{L}(X,X^*)$ when $g=0$. It follows that $\|A_g \mathbf{q}\|^2_X \leq c_A (\|\mathbf{q}\|^2_X+\|g\|^2_{L^2(\partial \Omega)})$ for all $\mathbf{q} \in X$ and for a fixed $g \in L^2(\partial \Omega)$ and some positive constant $c_A$. Let $\mathbf{v} \in X$ such that $\|\mathbf{v}\|_X \leq 1$. Then, using inequality \eqref{eq:prevest} and the constants $c_A$ and $c_i$, $i=1,2,3$, we find that any weak solution $\mathbf{y}$ of system \eqref{eq:BCS1} satisfies the following inequality,
		\begin{align*}
		\bigg| \langle \frac{\partial \mathbf{y}}{\partial t}(t),\mathbf{v} \rangle_{X^*,X} \bigg |   \leq ~C'' (\|\mathbf{y}(t)\|_{X}  + \|\mathbf{f}(t)\|_{L^2(\Omega)^3}  + \|g(t)\|_{L^2(\partial \Omega)}),
		\end{align*}
		for almost every $t \in (0,T)$ and for the constant $C''$   = $r \times C' \times  {\rm max} \lbrace c_A, c_1, c_2, c_3 \rbrace $, where $r = \max \limits_{1 \leq i \leq 3} |u^{max}_i|+\max \limits_{1 \leq i \leq 3} |u^{min}_i|+\max \limits_{i=1,2}|b_i|+1$. Here, we are bounding the left-hand side of equation \eqref{eq:WeakBCS} by setting $\boldsymbol{\phi}(t) = \boldsymbol{v}$ for almost every $t \in (0,T)$. The inequality above implies that
		\begin{align*}
		\int_0^T \| \frac{\partial \mathbf{y}}{\partial t}(t)\|^2_{X^*} dt 
		\leq ~ C''  \int_0^T  ( \|\mathbf{y}(t)\|^2_{X} + \|\mathbf{f}(t)\|^2_{L^2(\Omega)^3}  + \|g(t)\|^2_{L^2(\partial \Omega)})dt
		.
		\end{align*} 
		The bound on the term $\|\partial \mathbf{y} / \partial t \|_{L^2(0,T;X^*)}$ follows from the above inequality. Hence, the bounds in the energy estimate \eqref{eq:BCSenergy} hold for $K = {\rm max}\lbrace 2C', 2C''\rbrace$.
		
		To confirm the existence and uniqueness of the solution to the PDE \eqref{eq:Macro1}, we note that the weak form of the PDE \eqref{eq:WeakBCS} is in the abstract form
		\begin{equation}
		\frac{d \mathbf{y}}{dt} = \tilde{A}(t)\mathbf{y}+\mathbf{\tilde{f}}(t), ~~~~\mathbf{y}(0)=\mathbf{y}_0,
		\end{equation}
		where $\tilde{A}(t) = A_0(t)+\sum_{i=1}^{3}  u_{i}(t)B_{i} \in \mathcal{L}(X,X^*)$ and $\mathbf{\tilde{f}} \in L^2(0,T;X^*)$ is given by $ \langle \mathbf{\tilde{f}}(t), \boldsymbol{\psi} \rangle_{X,X^*} = \langle \mathbf{f}(t), \boldsymbol{\psi} \rangle_{L^2(\Omega)^3} + \int_{\partial \Omega} g(\mathbf{x},t) \psi_1(\mathbf{x})d\mathbf{x}$ for almost every $t \in [0,T]$. Here, the term $\int_{\partial \Omega} g(\mathbf{x},t) \psi_1(\mathbf{x})d\mathbf{x}$ defines a bounded functional on $X$ due to the bounds on the trace map $\tau$ \cite{grisvard2011elliptic}[Theorem 1.5.1.10]. Then the result follows from \cite{wloka1987partielle}[Theorem 26.1] and the estimate \eqref{eq:energy}.
	\end{proof}
	
	\subsection{Existence of Optimal Control} \label{sec:ExOpt}
	
	In this section, we prove the existence of a solution to the optimal control problem \eqref{eq:OptControlProblem}. Toward this end, we define the control-to-state mapping, $\Xi$: $U_{ad} \rightarrow \bar{Y}$, which maps a control, $\mathbf{u}$, to $\mathbf{y}$, the corresponding solution to system \eqref{eq:BCS1} for $\mathbf{f} = \mathbf{0}$ and $g=0$. We introduce the following reduced optimization problem,
	\begin{equation}
	\min\limits_{\mathbf{u} \in U_{ad}} \hat{J}(\mathbf{u}) := J(\Xi(\mathbf{u}),\mathbf{u}), \label{eq:redproblem}
	\end{equation}
	where $J$ is defined in problem \eqref{eq:OptControlProblem}.  This problem incorporates the PDE system \eqref{eq:BCS1} into the objective functional, rather than defining it as a set of constraints as in the original problem formulation \eqref{eq:OptControlProblem}. We shall henceforth analyze the reduced problem \eqref{eq:redproblem}.
	
	\vspace{2mm}
	
	\begin{theorem}
		An optimal control $\mathbf{u}^*$ exists that minimizes the objective functional $\hat{J}$ in problem \eqref{eq:redproblem}.
	\end{theorem}
	\begin{proof}
		The functional $\hat{J}(\mathbf{u})$ is bounded from below, which implies that the infimum is a finite real number.
		Therefore, $q = \inf_{\mathbf{u} \in U_{ad}} \hat{J}(\mathbf{u})$ exists. We now determine  an optimal pair $(\mathbf{y}^*,\mathbf{u}^*)$, for which $ J(\mathbf{y}^*,\mathbf{u}^*) = q$.  Let $\lbrace \mathbf{u}^n \rbrace^{\infty}_{n=1}$ be a minimizing sequence such that $\hat{J}(\mathbf{u}^n) \rightarrow q$ as $n \rightarrow \infty$.  $U_{ad}$ is a bounded and closed convex set, and thus is weak sequentially compact. 
		Hence, there exists a subsequence $\lbrace \mathbf{u}^n \rbrace^{\infty}_{n=1}$ such that
		\begin{equation}
		\mathbf{u}^n \rightharpoonup \mathbf{u}^* \hspace{5mm} in \hspace{2mm} L^2(0,T)^{3},
		\end{equation}
		with $\mathbf{u}^* \in U_{ad}$.  Recall from Section \ref{sec:mathprelim} that $X = V \times L^2(\Omega)^{2}$. Then, the uniform boundedness of the solution, $\mathbf{y}$, in the energy estimate \eqref{eq:BCSenergy} allows us to extract a subsequence $\lbrace \mathbf{y}^n \rbrace^{\infty}_{n=1}$, where $\mathbf{y}^n = \Xi(\mathbf{u}^n)$, such that
		\begin{equation}
		\mathbf{y}^n \rightharpoonup \mathbf{y}^* \hspace{5mm} in \hspace{2mm} L^2(0,T;X).
		\end{equation}
		
		It is necessary to confirm that $\Xi(\mathbf{u}^*) = \mathbf{y}^*$, since we do not know whether $\Xi$ is weakly continuous. Using the energy estimate \eqref{eq:BCSenergy}, we have uniform bounds on $\|\mathbf{y}^n\|_{L^2(0,T;X)} +\|\partial \mathbf{y}^n / \partial t \|_{L^2(0,T;X^*)}$. Therefore, by the Aubin-Lions lemma \cite{simon1986compact},  we have that there exists a subsequence such that
		\begin{equation}
		y_1^n \rightarrow y^*_1 \hspace{5mm} in \hspace{2mm} L^2(0,T;L^2(\Omega)).
		\end{equation}
		The energy estimate \eqref{eq:BCSenergy} implies that the terms $\nabla y^n_1$, $\frac{\partial \mathbf{y}^n}{\partial t}$, and $k_f y^n_2$ are uniformly bounded. Hence, we can also conclude that there exist subsequences that satisfy:
		\begin{align*}
		\nabla y^n_1 \rightharpoonup & ~\nabla  y^*_1  \hspace{5mm}in \hspace{2mm} L^2(0,T;L^2(\Omega)), \nonumber \\
		\nabla y^n_1 \rightarrow & ~\nabla y^*_1 \hspace{5mm} in \hspace{2mm} L^2(0,T;V^*), \nonumber \\
		\frac{\partial \mathbf{y}^n}{\partial t}  \rightharpoonup & ~\frac{\partial \mathbf{y}^*}{\partial t} \hspace{5mm}in \hspace{2mm} L^2(0,T;X^*), \nonumber \\
		k_f y^n_2 \rightharpoonup & ~k_f y^*_2  \hspace{5mm}in \hspace{2mm} L^2(0,T;L^2(\Omega)). 
		\end{align*}
		The first two components of $\mathbf{u}^n$ are denoted by $\mathbf{v}^n$ and the third component by $k^n$. 
		From the strong convergence of $y_1^n$ in $L^2(0,T;L^2(\Omega))$ and the weak convergence of $\mathbf{u}^n$ in $L^2(0,T)^{3}$, we can further deduce that: 
		\begin{align*}
		k^n H_\Gamma y^n_{1} \rightharpoonup & ~k H_\Gamma y^*_{1} \hspace{5mm} in \hspace{2mm} L^2(0,T;L^2(\Omega)), \nonumber \\
		\mathbf{v}^n \nabla y^n_1 \rightharpoonup & ~\mathbf{v} \nabla y^*_1 \hspace{5mm} in \hspace{2mm} L^2(0,T;V^*).
		\end{align*}
		To prove convergence of the terms in the weak form of the PDE \eqref{eq:WeakBCS} that arise from the boundary condition \eqref{eq:NFBC}, we apply Green's theorem 
		\begin{equation}
		\begin{split}
		\langle \mathbf{v}^n \cdot \nabla y^n_1, \phi \rangle_{L^2(0,T;L^2(\Omega))} + \int_0^T \int_{\partial \Omega}\mathbf{n} \cdot (\mathbf{v}^n y_1 \phi) d\mathbf{x} dt \\
		= ~ -\langle \mathbf{v}^n \cdot y^n_1 , \nabla \phi \rangle_{L^2(0,T;L^2(\Omega))}
		\end{split}
		\end{equation}
		for all $\mathbf{\phi} \in L^2(0,T;V)$. Due to the strong convergence of $y^n_1$ in $L^2(0,T;L^2(\Omega))$ and the weak convergence of $\mathbf{v}^n$ in $L^2(0,T)^2$, we have that
		\begin{equation}
		\mathbf{v}^n \cdot y^n_1  \rightharpoonup  \mathbf{v}^* \cdot y_1^* \hspace{5mm} in \hspace{2mm} L^2(0,T;L^2(\Omega)). 
		\end{equation}
		The convergence of all the sequences defined thus far implies that the sequence of solutions $\mathbf{y}^n = \Xi(\mathbf{u}^n)$, which satisfies
		\begin{equation}
		\langle \frac{\partial \mathbf{y}^n}{\partial t},\phi \rangle_{Y^*,Y} = \langle A_g \mathbf{y}^n,\phi \rangle_{Y^*,Y} + \sum_{i=1}^{3}  \langle \mathbf{u}^n_{i}B_{i}\mathbf{y}^n,\phi \rangle_{F} \label{eq:existBCS}
		\end{equation}
		with $g=0$, converges to the solution $\mathbf{y}^* = \Xi(\mathbf{u}^*)$, which satisfies
		\begin{equation}
		\langle \frac{\partial \mathbf{y}^*}{\partial t},\mathbf{\phi} \rangle_{Y^*,Y} = \langle A_g\mathbf{y}^*,\mathbf{\phi} \rangle_{Y^*,Y} + \sum_{i=1}^{3}  \langle \mathbf{u}^*_{i}B_{i}\mathbf{y}^*,\mathbf{\phi} \rangle_{F} \label{eq:existBCSc}
		\end{equation}
		with $g=0$.
		
		It remains to be shown that $\hat{J}(\mathbf{u}^*)=q$. Because $J$ is weakly lower semicontinuous, we can state that
		\begin{equation}
		q = \lim_{n \rightarrow \infty}J(\mathbf{y}^n,\mathbf{u}^n) ~\leq~ J(\mathbf{y}^*,\mathbf{u}^*).
		\end{equation}
		Since $q$ was defined earlier as the infimum of $\hat{J}(\mathbf{u})$, we conclude that
		\begin{equation}
		\hat{J}(\mathbf{u}^*)= J(\mathbf{y}^*,\mathbf{u}^*) = q,
		\end{equation}
		which completes the proof.
	\end{proof}
	
	\subsection{Differentiability of the Control-to-State Map}\label{sec:diff}
	In this section, we summarize several results from our prior work \cite{elamvazhuthi2015optimal} that we apply to derive the gradient of $\hat{J}(\mathbf{u})$. We use the gradient in the method of projected gradient descent to numerically compute a locally optimal control. The numerical computation of  this gradient is made possible by deriving an expression for the gradient in terms of the {\it adjoint equation}, which is a system of PDEs. This method for computing the control variables is explained in \cite{elamvazhuthi2014variational}[Section 5.2]. See also \cite{troltzsch2010optimal}[Section 3.7.1], where the method is discussed in the more general context of optimal control of parabolic PDEs.
	
	
	\vspace{2mm}
	
	\begin{proposition}
		The mapping $\Xi$ is directionally differentiable at every $\mathbf{u} \in U_{ad}$ along each direction $\mathbf{h} \in L^{\infty}(0,T)^{3}$. Its directional derivative $\Xi'(\mathbf{u}) : U_{ad} \rightarrow \bar{Y}$ in the direction $\mathbf{h} \in L^{\infty}(0,T)^3$, denoted by $\Xi'(\mathbf{u})\mathbf{h}$, is given by the solution $\mathbf{w}$ to the following equation,
		\begin{eqnarray}
		&& \frac{\partial \mathbf{w}}{\partial t} = A\mathbf{w} + \sum_{i=1}^{3} u_{i}B_{i}\mathbf{w}  + \sum_{i=1}^{3} h_{i}B_{i}\mathbf{y}  \hspace{2mm}  in \hspace{2mm} Q,\nonumber \\
		&& \mathbf{n} \cdot(\nabla w_1 - [u_1 ~ u_2]^T \cdot w_1) = \mathbf{n} \cdot ( [h_1 ~ h_2]^T y_1) \hspace{2mm} on \hspace{2mm} \Sigma,\nonumber \\
		&& \mathbf{w}(0) = 0  \hspace{2mm} on \hspace{2mm} \Omega.
		\end{eqnarray}
	\end{proposition}
	
	\vspace{2mm}
	
	\begin{theorem}
		The reduced objective functional $\hat{J}$ is directionally differentiable at every $\mathbf{u} \in U_{ad}$ along all $\mathbf{h} \in L^{\infty}(0,T)^{3}$. Its directional derivative along $\mathbf{h} \in  L^{\infty}(0,T)^3$ has the form 
		\begin{equation}
		\begin{split}
		d\hat{J}(\mathbf{u})\mathbf{h}  = &\int^T_0 \langle  \mathbf{n} \cdot ( [h_1 ~ h_2]^T p_1),y_1 \rangle_{L^2(\partial \Omega)} dt
		\\  + & \int^T_0 \langle \sum_{i=1}^{3} h_i B_i \mathbf{y},\mathbf{p} \rangle_{L^2(\Omega)^{3}}dt + \lambda \langle \mathbf{u},\mathbf{h} \rangle_{L^2(0,T)^{3}},
		\label{eq:objpan}
		\end{split}
		\end{equation} 
		where $\mathbf{p}$ is the solution of the backward-in-time adjoint equation,
		\begin{eqnarray}
		&& \hspace{-2mm} -\frac{\partial p_{1}}{\partial t} = \nabla \cdot (D \nabla p_{1}+ \mathbf{v}(t) p_{1}) + k(t) H_{\Gamma} (-p_{1}+p_{2}+p_{3})  \hspace{2mm}  \nonumber \\
		&& \hspace{1.2cm} in \hspace{2mm} Q, \nonumber \\
		&&  \hspace{-2mm} -\frac{\partial p_{2}}{\partial t} =  k_{f}p_1-k_{f}p_2 \hspace{2mm} ~~in \hspace{2mm} Q,\nonumber \\
		&& \hspace{-2mm}  -\frac{\partial p_{3}}{\partial t} = 0 \hspace{2mm} ~~in \hspace{2mm} Q, 
		\label{eq:AdjointMacro1}
		\end{eqnarray}
		with the Neumann boundary conditions
		\begin{equation}
		\mathbf{n} \cdot \nabla p_{1}= 0 \hspace{2mm} ~~ on \hspace{2mm} \Sigma
		\label{eq:adjbc}
		\end{equation} 
		and final time condition
		\begin{equation}
		\mathbf{p}(T) = W^*(W\mathbf{y}(\cdot,T)-\mathbf{y}_{\Omega}),
		\label{eq:adjic}
		\end{equation} 
		where $W$ is defined in Problem \ref{coveragetask_meanfield}.
	\end{theorem}
	
	\vspace{2mm}
	\begin{remark}
		The adjoint equation \eqref{eq:AdjointMacro1}-\eqref{eq:adjic} should not be confused with the backward heat equation, which is known to be ill-posed in $H^1(\Omega)$. The well-posedness of system \eqref{eq:AdjointMacro1}-\eqref{eq:adjic} can be established using a standard variable transformation of the form $\mathbf{p}^*(t) = \mathbf{p}(T-t)$. The resulting equation,
		\begin{eqnarray}
		&& \hspace{-2mm} \frac{\partial p^*_{1}}{\partial t} = \nabla \cdot (D \nabla p^*_{1}+ \mathbf{v}(t) p^*_{1}) + k(t) H_{\Gamma} (-p^*_{1}+p^*_{2}+p^*_{3}) \nonumber \\
		&& \hspace{1cm} in \hspace{2mm} Q, \nonumber \\
		&&  \hspace{-2mm} \frac{\partial p^*_{2}}{\partial t} =  k_{f}p^*_1-k_{f}p^*_2 \hspace{2mm} ~~ in \hspace{2mm} Q,\nonumber \\
		&& \hspace{-2mm}  \frac{\partial p^*_{3}}{\partial t} = 0 \hspace{2mm}~~ in \hspace{2mm} Q, 
		\label{eq:CAdjointMacro1}
		\end{eqnarray}
		with the Neumann boundary conditions
		\begin{equation}
		\mathbf{n} \cdot \nabla p^*_{1}= 0 \hspace{2mm} ~~on \hspace{2mm} \Sigma
		\label{eq:cadjbc}
		\end{equation} 
		and initial condition
		\begin{equation}
		\mathbf{p}^*(0) = W^*(W\mathbf{y}(\cdot,T)-\mathbf{y}_{\Omega}),
		\label{eq:cadjic}
		\end{equation}
		can be shown to have a unique solution $\mathbf{p}^*$ using arguments similar to those in the proof of Lemma \ref{ExistLemma}. 
	\end{remark}
	
	\section{Simulation Results and Discussion}
	\label{sec:simulat}
	
	In this section, we validate our approaches to the mapping assignment (Section \ref{sec:optMap}) and coverage assignment (Section \ref{sec:optCoverage}) with numerical simulations. 
	The PDE models of ensemble mapping and coverage, presented in Section \ref{sec:macrosec}, are numerically solved using the {\it method of lines}. In this method, the operators associated with the PDE are initially discretized in space. For this spatial discretization, we use a finite-difference and flux limiter based scheme on a uniform mesh. Time is viewed as a continuous variable, and the resulting system of ordinary differential equations (ODEs) is solved using built-in ODE solvers in MATLAB. See \cite{hundsdorfer2013numerical} for further details on this approach to numerically solving advection-diffusion-reaction type PDEs.
	
	
	\subsection{Mapping Assignment} \label{sec:mappingtasksim}
	
	We validate our mapping approach in two test cases.  In both scenarios, the domain is $\Omega = [0,100]^2~m^2$, the ensemble has $N = 30$ agents, the diffusion coefficient is $D = 10^{-4}~m^2/s$, and the agents' rate of recording observations over a region of interest is $k_o = 100~s^{-1}$, yielding a high probability rate of registering observations.  The agents' velocity field $\mathbf{v}(t)$ is defined to drive the agents along a lawnmower path, as illustrated by the two agent trajectories in \autoref{fig:AM}. This path is perturbed by stochastic fluctuations arising from the agents' diffusive motion.  
	\begin{figure}
		\centering
		\includegraphics[trim = 45mm 68mm 50mm 68mm, scale=0.4]{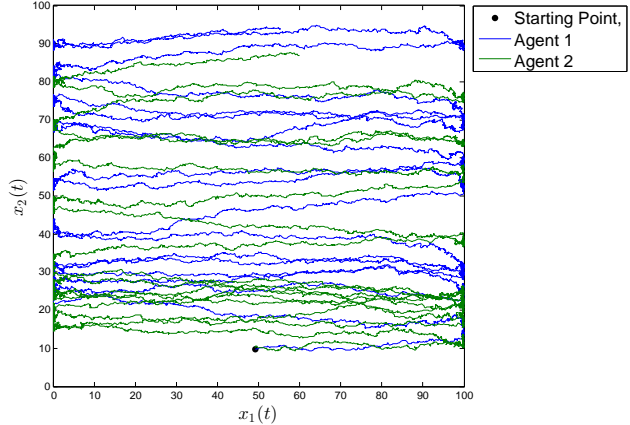}
		\caption{Trajectories of two agents during the mapping assignment.}
		\label{fig:AM}  
	\end{figure}  
	
	We define two environments, {\it Case 1} and {\it Case 2}, with different regions of interest, labeled {\it Actual $H_\Gamma(\mathbf{x})$} in \autoref{fig:Maps}.  The results of our mapping approach are shown in the plots labeled {\it Estimated $H_\Gamma(\mathbf{x})$}.  To use the estimated maps in simulations of our coverage approach (Section \ref{sec:covtasksim}), we thresholded the estimated $H_\Gamma(\mathbf{x})$ in both environments by defining the {\it Thresholded $H_\Gamma(\mathbf{x})$}, also shown in \autoref{fig:Maps}, as $H_{T}(\mathbf{x}) =1$ for each $\mathbf{x} \in \Omega$ such that $H_\Gamma(\mathbf{x}) \geq 0.5$ and $H_{T}(\mathbf{x})=0$ otherwise.  In both cases, our optimization method is able to reconstruct the spatial coefficient $H_\Gamma(\mathbf{x})$ with considerable accuracy, even though a relatively small number of robots was used ($N=30$).  The largest error in the estimates occurs in the top half of the  {\it Case 2} environment, which can be attributed to the increased dispersion of the agents, due to their diffusive motion, as they reach the upper portion of the domain. We would expect a larger ensemble of agents to generate a more accurate map, since the microscopic model converges to the macroscopic model as $N \rightarrow \infty$ \cite{zhang2017performance}.
	
	
	\begin{figure}
		\centering
		\subfigure[Case $1$]
		{      
			\includegraphics[trim = 70mm 75mm 60mm 80mm, scale=0.74]{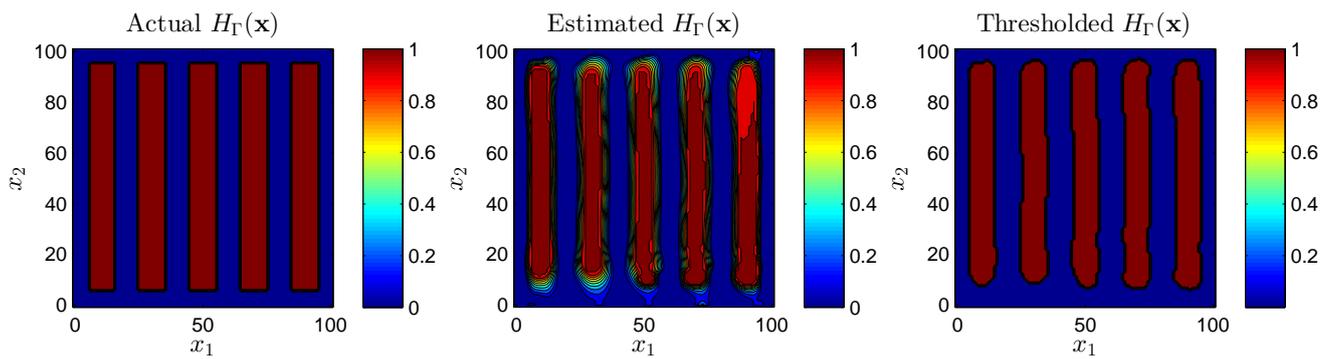}
			\label{fig:Map1} 
		}\\
		
		\subfigure[Case $2$]
		{      
			\includegraphics[trim = 70mm 75mm 60mm 70mm, scale=0.74]{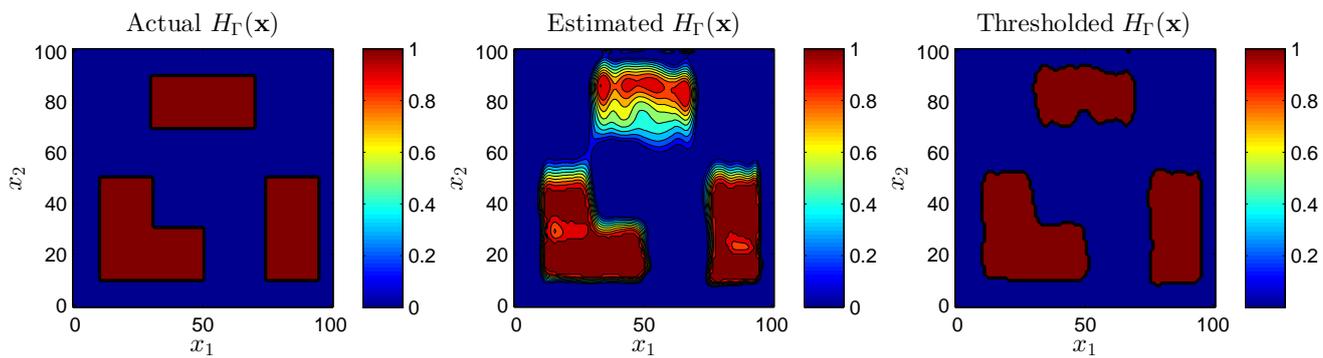}
			\label{fig:Map2} 
		}\\
		
		\caption{Actual, estimated, and thresholded maps of regions of interest.}
		\label{fig:Maps} 
		\vspace{15mm}
	\end{figure}

	\subsection{Coverage Assignment} \label{sec:covtasksim}
	
	We validate our coverage approach in three test cases, {\it Cases 1, 2,} and {\it 3}. The region of interest in {\it Case 1} is defined as {\it Thresholded $H_\Gamma(\mathbf{x})$} in Figure \autoref{fig:Map1}, and the region of interest in {\it Cases 2} and {\it 3} is defined as {\it Thresholded $H_\Gamma(\mathbf{x})$} in Figure \autoref{fig:Map2}.
	In all scenarios, the domain is $\Omega = [0,100]^2~m^2$, and the diffusion coefficient is $D = 5 \times 10^{-4} ~ m^2/s$. To illustrate the scalability of our control methodology to large numbers of agents, we simulate an ensemble with $N = 1000$ agents. All agents are initialized in the {\it moving} state, with an initial density given by a Gaussian distribution centered at $[x_{o1}~ x_{o2}] = [10 ~10]$. Thus, the initial conditions \eqref{eq:covIC} are:
	\begin{eqnarray}
	y_1(\mathbf{x},0) &=& A \exp \left(-\frac{(x_1-x_{o1})^2}{2\sigma_{x_1}^2}-\frac{(x_2-x_{o2})^2}{2\sigma_{x_2}^2}\right), \nonumber \\ 
	y_{2}(\mathbf{x},0) &=& 0,  \nonumber \\ 
	y_{3}(\mathbf{x},0) &=& 0,  \label{eq:MacroSi}
	\end{eqnarray}
	with $\sigma_{x_1} = \sigma_{x_2} =0.02$ and $A$ defined such that the Gaussian distribution integrates to $1$ over the domain.  The final time is set to $T = 800~s$ for {\it Cases 1} and {\it 2} and $T = 300~s$ for {\it Case 3}.
	
	We define $\mathbf{y}_\Omega$, the target spatial distribution in Problem \ref{coveragetask_meanfield}, as follows.  We partition the domain $\Omega$ into $P=20$ cells, denoted by $\{\Omega_{nm}\}$.  The cell $\Omega_{00}$ occupies the region $[0 , \frac{100}{P}] \times [0 , \frac{100}{P}]$, and all other cells are defined as $\Omega_{nm} = \big(100 \frac{m}{P},100 \frac{m+1}{P} \big] \times (100 \frac{n}{P},100 \frac{n+1}{P}]$, where $n,m \in \lbrace 0,1, ..., P-1 \rbrace$ and $n\neq$ 0 or $m\neq0$.  For {\it Case 1} and {\it Case 2}, we set the target number of instances of desired robot activity in each cell $\Omega_{mn}$ to be $z_{mn}^* = C \times \frac{\mu(\Omega_{nm} \hspace{1mm} \cap \hspace{1mm} \Gamma)}{\mu(\Omega_{mn})}$, where $C$ is a positive constant and $\mu$ is the Lebesgue measure on $\Omega$.  The target distribution of robot coverage activity, which is the third component of  $\mathbf{y}_\Omega$, is defined as $y^*_3(\mathbf{x},T) = \frac{C}{50}$ for all $\mathbf{x} \in \Gamma$.  Since we do not require the first two components of $\mathbf{y}_\Omega$ (the densities of {\it moving} and {\it stationary} robots) to reach target distributions at time $T$, we set the function $W$ in Problem \ref{coveragetask_meanfield} to be a diagonal operator matrix of the form $diag([0 \hspace{2mm} 0 \hspace{2mm}  I])$, where $I$ is the identity operator on $L^2(\Omega)$. 
	We specify the following sub-cases: $C = 450$ in {\it Case 1a} and {\it Case 2a}, and $C = 3600$ in {\it Case 1b} and {\it Case 2b}.  For {\it Case 3}, coverage activity is desired only in the upper half of the domain; we set $y^*_3(\mathbf{x},T) = 36$ for $\mathbf{x} \in \Gamma$ with $x_2 \geq 60$, and $y^*_3(\mathbf{x},T) = 0$ otherwise. 
	
	

	
	
	%
	%
	
	\autoref{fig:OBJfig} plots the time evolution of the objective function $J$ for each case.  For {\it Cases 1a,b} and {\it 2a,b}, the low values of $J$ at the final time $T=800~s$ indicate that the target coverage density is nearly achieved.  This is in part due to the accuracy of the thresholded estimate of $H_\Gamma(\mathbf{x})$ obtained from the mapping task (Section \ref{sec:mappingtasksim}). The lower error in the estimated $H_\Gamma(\mathbf{x})$ for {\it Case 1} compared to this mapping error for {\it Case 2} contributes to the better coverage performance (i.e., lower values of $J$) in {\it Case 1a} versus {\it Case 2a} and in {\it Case 1b} versus {\it Case 2b}.  
	For {\it Case 3}, the value of $J$ is higher at the final time $T=300~s$ than the values of $J$ for {\it Cases 1a,b} and {\it 2a,b} at $T=800~s$, indicating that the final coverage density is relatively farther from the target density. 
	The poorer coverage performance in {\it Case 3}, in which coverage activity is limited to a subset of the region of interest, is a consequence of the limited controllability of the system, which can be attributed to three factors.  First, only three control variables are used to control the PDE model, which is an infinite-dimensional dynamical system. Second, the system is constrained to achieve the target coverage density by time $T$.  Third, the time-dependent diffusion, reaction, and advection operators in the PDE model commute \cite{lanser1999analysis}, and this  commutativity of the control and drift vector fields degrades the controllability properties of the system \cite{agrachev2013control}. Thus, the ensemble would achieve better coverage performance in assignments with less stringent requirements on the system controllability properties. Such assignments include those in which the target coverage distribution is proportional to an environmental parameter, or in which the objective is to achieve a minimum coverage density in each region of interest rather than an exact coverage density.

	Figure \ref{fig:CovCases} plots the target coverage density $z^*_{mn}$ for {\it Cases 1b} and {\it 2a} alongside the corresponding expected coverage density, $y_3(\mathbf{x},T)$ from the macroscopic PDE model \eqref{eq:BCS1}, and the achieved coverage density, $z^{\Omega_{mn}}_3(T)$ 
	from the microscopic model, at the final time $T=800$ s. The lower value of the target coverage density in {\it Cases 1a, 2a} ($C=450$) than in {\it Cases 1b, 2b} ($C=3600$) results in larger stochastic fluctuations of the achieved coverage density around the expected coverage density. These larger fluctuations produce poorer coverage performance (higher values of $J$) in {\it Case 1a} versus {\it Case 1b} and in {\it Case 2a} versus {\it Case 2b}.  
	The plots in Figure \ref{fig:CovCases} display the relative degree of these fluctuations in {\it Case 1b} versus {\it Case 2a} at time $T=800~s$. The plots show that the achieved coverage density from the microscopic model approximates the expected coverage density from the macroscopic model, due to the large number of agents ($N=1000$) in the ensemble. As demonstrated in \cite{zhang2017performance}, the discrepancy between these two density fields will tend to zero as $N \rightarrow \infty$.

	

	\begin{figure}
		
		\centering
		
		\includegraphics[trim =21mm 10mm 10mm 10mm, scale=0.39]{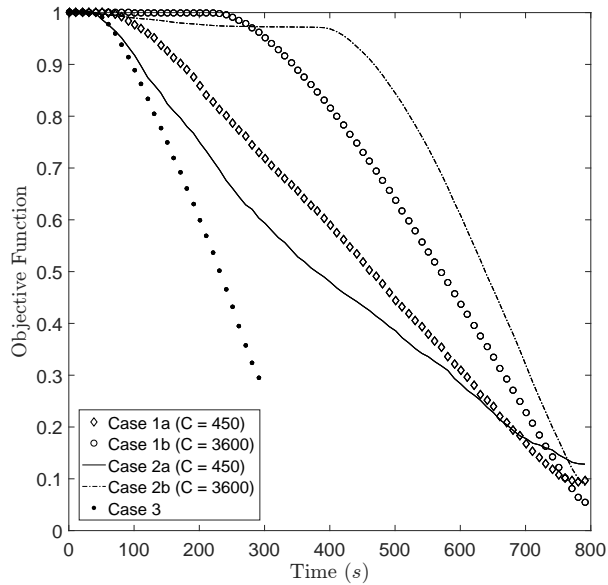}
		
		\caption{Time evolution of the objective function $J$ in Problem \ref{coveragetask_meanfield} for different coverage scenarios.}
		
		\label{fig:OBJfig}   
	\end{figure}

	\begin{figure*}
		\centering
		\subfigure[Case $1b$]
		{      
			\includegraphics[trim = 70mm 70mm 60mm 70mm, scale=0.74]{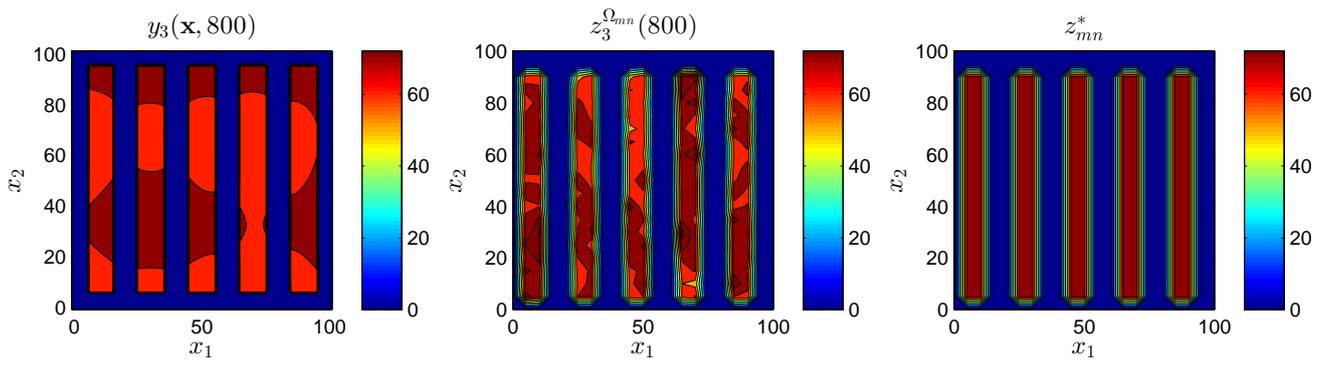}
			\label{fig:Case1b}  
		}\\
		
		\subfigure[Case $2a$]
		{      
			\includegraphics[trim = 70mm 70mm 60mm 70mm, scale=0.74]{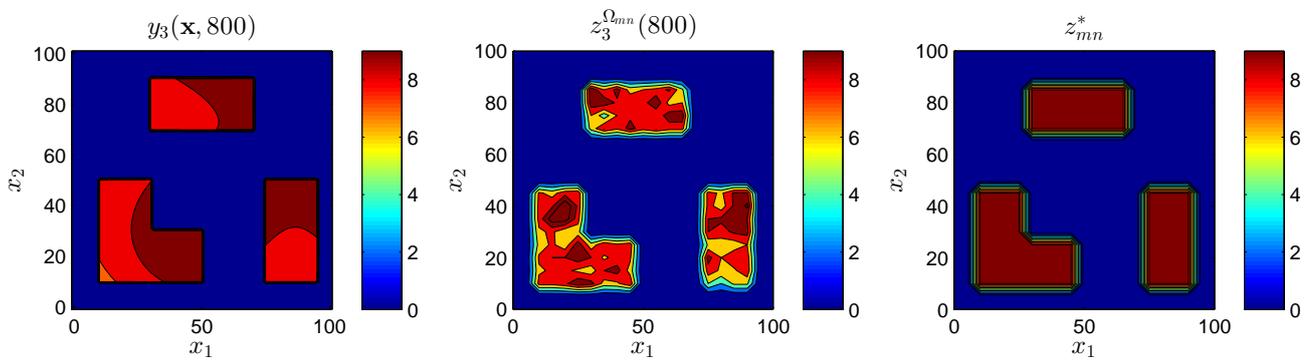}
			\label{fig:Case2a}  
		}\\
		\caption{Expected, achieved, and target densities of robot coverage activity for two scenarios.}
		\label{fig:CovCases}  
	\end{figure*}

	\section{Conclusion} \label{sec:conc}
	In this paper, we presented stochastic approaches to mapping and coverage assignments for an ensemble of autonomous robots in a PDE control framework. This framework enables the modeling and control of a robotic ensemble in a rigorous way that is scalable with the number of robots. We demonstrated that temporal data obtained by a small ensemble of diffusive agents can provide rich information about the spatial distribution of a region of interest, despite severe restrictions on the agents' sensing, localization, tracking, and computational capabilities.  We also showed that we can pose the coverage task as an optimal control problem that computes the agents' control inputs to achieve a target distribution of coverage activity over the previously mapped regions of interest.  
	
	In future work, we will investigate the controllability properties of the PDE models that we have presented in this paper. Additionally, we plan to incorporate pairwise interactions between agents, such as those defined by attraction-repulsion potentials, in order to increase the cohesiveness of the ensemble and improve the reachability properties of the system.
	The component of the robots' velocity field that is induced by pairwise interactions would be included in the advection term, resulting in a nonlinear PDE as the macroscopic model.  This type of model would require more advanced analytical tools than the ones used in this paper.
	
	
	
	
	
	
	
	
	\addcontentsline{toc}{part}{REFERENCES}
	\bibliographystyle{plain}        
	
	\bibliography{autosam2}  

\end{document}